\documentclass[11pt]{article}
\usepackage[english]{babel}
\usepackage{amsmath}
\usepackage{amsfonts}
\usepackage{amssymb}
\usepackage{mathrsfs}
\usepackage{amsthm}
\usepackage{graphicx}
\usepackage{comment}

\theoremstyle{definition}
\newtheorem{definition}{Definition}[section]

\theoremstyle{remark}
\newtheorem{remark}{Remark}[section]

\theoremstyle{plain}
\newtheorem{theorem}{Theorem}[section]
\newtheorem{proposition}{Proposition}[section]
\newtheorem{corollary}{Corollary}[section]
\newtheorem{lemma}{Lemma}[section]

\numberwithin{equation}{section}

\def\be{\begin{equation}}
\def\ee{\end{equation}}

\newcommand{\E}{\mathbb{E}}
\newcommand{\R}{\mathbb{R}}
\newcommand{\C}{\mathbb{C}}
\newcommand{\Rpz}{\R_{+}}

\newcommand{\qp}{q_{L_p}}
\newcommand{\qpp}{q_{L_{p+1}}}

\newcommand{\va}{a}
\newcommand{\vlambda}{\lambda}
\newcommand{\betap}{\beta_p}
\newcommand{\betapm}{\beta_{p-1}}
\newcommand{\betauno}{\beta_1}
\newcommand{\betaKm}{\beta_{K-1}}
\newcommand{\tp}{t_p}

\DeclareMathOperator{\DDelta}{\Delta}
\DeclareMathOperator{\pSK}{p^\textup{SK}}
\DeclareMathOperator{\qSK}{q^\textup{SK}}
\DeclareMathOperator{\qSKRS}{q^\textup{RS-SK}}

\DeclareMathOperator{\pSKN}{p^\textup{SK}_N}
\DeclareMathOperator{\pSKNp}{p^\textup{SK}_{N_p}}
\DeclareMathOperator{\pASK}{p^\textup{SK-A}}

\DeclareMathOperator{\PDBM}{\mathcal{P}^\textup{DBM}}
\DeclareMathOperator{\pDBMN}{p^\textup{DBM}_{\Lambda_N}}
\DeclareMathOperator{\pADBM}{p^\textup{DBM-A}}
\DeclareMathOperator{\pRS}{\mathcal{P}^\textup{RS-DBM}}
\DeclareMathOperator{\pSKRS}{\mathcal{P}^\textup{RS-SK}}
\DeclareMathOperator{\pRSN}{\mathcal{P}^\textup{RS-DBM}_{\Lambda_N}}

\DeclareMathOperator{\Jac}{Jac}
\DeclareMathOperator{\diag}{diag}
\DeclareMathOperator{\sign}{sign}

\newcommand{\HSKN}{H^\textup{SK}_N}
\newcommand{\HSKp}{H^\textup{SK}_{L_p}}
\newcommand{\HSKu}{H^\textup{SK}_{L_1}}
\newcommand{\HSKK}{H^\textup{SK}_{L_K}}
\newcommand{\ZSKN}{Z^\textup{SK}_N}

\begin{document}

\title{Deep Boltzmann machines: rigorous results at arbitrary depth.}

\author{Diego Alberici, Pierluigi Contucci, Emanuele Mingione}

\maketitle

\begin{abstract}
A class of deep Boltzmann machines is considered in the simplified framework of a quenched system with Gaussian noise and independent entries. The quenched pressure of a $K$-layers spin glass model is studied allowing interactions only among consecutive layers. A lower bound for the pressure is found in terms of a convex combination of $K$ Sherrington-Kirkpatrick models and used to study the annealed and replica symmetric regimes of the system. A map with a one dimensional monomer-dimer system is identified and used to rigorously control the annealed region at arbitrary depth $K$ with the methods introduced by Heilmann and Lieb. The compression of this high noise region displays a remarkable phenomenon of localisation of the processing layers.
Furthermore a replica symmetric lower bound for the limiting quenched pressure of the model is obtained in a suitable region of the parameters and the replica symmetric pressure is proved to have a unique stationary point.
\end{abstract}


\section{Introduction and results}

The mean-field setting in Statistical Mechanics corresponds to the invariance of an $N$ particles system under the permutation group action. When this condition is weakened to permutation invariance within each set of a $K$-partition of the system $\big(\sum_{p=1}^{K}N_p=N\big)$, a homogeneous model generalizes to its $K$-populated version.
This generalization has been considered in spin systems for both non-random interactions, i.e. the Curie-Weiss model \cite{CGallo,CFedele}, and random interactions, i.e. the Sherrington-Kirkpatrick model \cite{BCMT,PanchenkoMSK}.
For the first case a complete control of the thermodynamic properties has been reached for general values of the interaction parameters. In the random case instead only the so called elliptic structure of the interactions is fully controlled, while the hyperbolic one is still not understood. We mention that the case $K=2$ has already been solved in two  particular frameworks characterized by replica symmetry: on the Nishimori line \cite{BMM} or with spherical spins \cite{Auchen,baik}. 

In this paper we continue the analysis started in \cite{ABCM,bipartiti} concerning a mean-field spin glass with pure hyperbolic structure of the interactions, i.e. a random version of deep Boltzmann machines [DBM] over $K$ layers \cite{Hinton1}. 
The framework of \cite{ABCM} is generalized by dealing with a general number $K$ of layers and by allowing local (layer dependent) temperatures.
A lower bound for the quenched pressure in terms of $K$ Sherringhton-Kirkpatrick models [SK] coupled in temperature along a linear chain is obtained and used to study the annealed and replica symmetric regimes of the random DBM in the large volume limit.

Our first result is a control of the annealed region $A_K$ in terms of the largest zero of a matching polynomial which -up to a change of variable in the complex plane- is the partition function of a monomer-dimer system over the linear chain of length $K$ \cite{HL,HLprl}. This region $A_K$ turns out to be exactly the one where the annealed solution $q=0$ is stable for the replica symmetric consistency equation.
The compression of the annealed region leads to a peculiar structure of the layers: in particular the extensive layers are localized along a chain of length two or three.

A replica symmetric lower bound for the quenched pressure is obtained in a suitable region of the parameters. In the case of Gaussian external fields this region is identified by a $K$-dimensional version of the Almeida-Thouless condition for SK. Within this framework the replica symmetric consistency equation is proved to have a unique solution on the whole space of parameters. It is important to mention that the uniqueness for the elliptic case \cite{PanchenkoMSK} is still an open problem.

The paper is organised as follows.
Section \ref{definizioni} introduces the model.
In Section \ref{lb} we provide a lower bound for the quenched pressure of the DBM in terms of an interacting variational principle. 
In Section \ref{sec:ann} we identify and study a region where the quenched and the annealed pressure of the DBM coincide.
In Section \ref{Sec:RS} we derive the replica symmetric functional for the DBM and we study its stationary point(s). 
In Section \ref{sec:rs bound} we provide a lower bound for the quenched pressure of the DBM in terms of the previous replica symmetric functional under suitable conditions on the parameters of the model. 
Appendix \ref{sec:appendix} contains properties of the matching polynomials zeros, which are useful to characterize the annealed region in Section \ref{sec:ann} and are mainly due to Heilmann and Lieb \cite{HL}.

\section{Definitions} \label{definizioni}
Consider $N$ spin variables $\sigma=(\sigma_i)_{i=1,\dots,N}\in\{-1,1\}^N$ arranged over $K$ layers $L_1,\dots,L_K$ of cardinality $N_1,\dots,N_K$ respectively, so that $\sum_{p=1}^K N_p=N\,$.
Assume that the relative sizes of the layers converge in the large volume limit:
\be\label{lambdadef}
\lambda_p^{(N)} \equiv\, \frac{N_p}{N} \,\xrightarrow[N\to\infty]{}\, \lambda_p \,\in[0,1]
\ee
for every $p=1,\dots,K\,$.
We denote $\Lambda_N=(L_p)_{p=1,\dots,K}\,$, 
$\lambda^{(N)}=\big(\lambda_p^{(N)}\big)_{p=1,\dots,K}$ and $\vlambda=(\lambda_p)_{p=1,\dots,K}$. Clearly $\sum_{p=1}^K\lambda_p=1\,$.


Let $J_{ij}$ for $(i,j)\in L_p\times L_{p+1}$ and $p=1,\dots,K-1\,$ be a family of i.i.d. standard Gaussian random variables coupling spins in two consecutive layers.
We introduce a vector of positive inverse temperatures tuning the interactions among consecutive layers $\beta=(\betap)_{p=1,\dots,K-1}\in\R_+^{K-1}\,$.

Let $h_i$ for $i\in L_p$ and $p=1,\dots,K$ be a family of independent real random variables, independent also of the $J_{ij}$'s, acting as external fields on the spins. Assume that $(h_i)_{i\in L_p}$ are i.i.d. copies of a random variable $h^{(p)}$ such that $\E|h^{(p)}|<\infty\,$. We denote $h=(h^{(p)})_{p=1,\dots,K}\,$.

\begin{definition}
The Hamiltonian of the random Deep Boltzmann Machine [DBM] is
\be \label{eq:H}
H_{\Lambda_N}(\sigma) \,\equiv\, -\frac{\sqrt{2}}{\sqrt{N}}\; \sum_{p=1}^{K-1}\,\betap \!\!\sum_{(i,j)\in L_p\times L_{p+1}}\!\!\!\!\! J_{ij}\, \sigma_i\sigma_j
\ee
for every spin configuration $\sigma\in\{-1,1\}^N\,$.
\end{definition}

\begin{definition}
Given two spin configurations $\sigma,\tau\in\{-1,1\}^N$, for every $p=1,\ldots,K$ we define their overlap over the layer $L_p$ as
\be
\qp(\sigma,\tau) \,\equiv\, \frac{1}{N_p}\,\sum_{i\in L_p} \sigma_i\,\tau_i \;\in[-1,1] \;.
\ee
\end{definition}

\begin{remark}
The covariance matrix of the centred Gaussian process $H_{\Lambda_N}$ is
\be\label{eq:cov H}
\E \,H_{\Lambda_N}(\sigma)\, H_{\Lambda_N}(\tau) \,=\,
N\, q_{\Lambda_N}(\sigma,\tau)^T\, M_1^{(N)}\, q_{\Lambda_N}(\sigma,\tau)
\ee
for every $\sigma,\tau\in\{-1,1\}^N$. Here we set $q_{\Lambda_N}(\sigma,\tau) \equiv \big(q_{L_p}(\sigma,\tau)\big)_{p=1,\dots,K}\;$,
\be \label{M1}
M_1(\beta,\lambda) \,\equiv\, \diag(\lambda)\,M_0(\beta)\, \diag(\lambda)\;,\vspace{-4pt}
\ee
\be \label{M0}
M_0(\beta) \,\equiv\,
\begin{pmatrix}
0 					& \beta_1^2				& 	 					&					& 				    &               \\[2pt]
\beta_1^2			& 0 					& \beta_2^2				& 					& 				    &               \\[2pt]
					& \beta_2^2				& 0						&  					& 				    &               \\
					& 						& 						& \ddots 			& 				    &               \\
					&						&						&					&				    & \beta_{K-1}^2 \\
					&						&						&                   & \beta_{K-1}^2\!   & 0 \\
\end{pmatrix}
\ee
and we denote $M_1^{(N)}\equiv M_1(\beta,\lambda^{(N)})\,$.
Notice that $M_0(\beta)$ can be interpreted as a weighted adjacency matrix for the layers structure of the DBM.
\end{remark}


\begin{definition}
The random partition function of the model introduced by Hamiltonian (\ref{eq:H}) is
\be
Z_{\Lambda_N} \,\equiv\, \sum_{\sigma\in\{-1,1\}^N} \exp\Bigg(-H_{\Lambda_N}(\sigma)\,+\,\sum_{p=1}^K\sum_{i\in L_p}h_i\,\sigma_i\Bigg)
\ee
and its quenched pressure density is
\be\label{qpres}
\pDBMN \,\equiv\, \frac{1}{N}\,\E\,\log Z_{\Lambda_N}
\ee
where $\E$ denotes the expectation over all the couplings $J_{ij}\,$'s and the external fields $h_i$'s.
\end{definition}


\section{A lower bound for the quenched pressure of the DBM}\label{lb}

In this section we give an explicit bound for the quenched pressure of the $K$ layers DBM in terms of $K$ independent Sherrington-Kirkpatrick spin-glasses [SK] \cite{Panchenko-Book,Tala,CG}.

Considering $N$ spin variables $\sigma_i$, $i=1,\dots,N$, we recall that the Hamiltonian of the SK model is
\be \label{eq:HSK}
\HSKN(\sigma) \,\equiv\, -\frac{1}{\sqrt{N}}\; \sum_{i,j=1}^N \tilde J_{ij}\, \sigma_i\sigma_j
\ee
where $\tilde J_{ij}$, $i,j=1,\dots,N$ is a family of i.i.d. standard Gaussian random couplings.
Given two spin configurations $\sigma,\tau\in\{-1,1\}^N$, their overlap is
\be
q_N(\sigma,\tau) \,\equiv\, \frac{1}{N}\,\sum_{i=1}^N \sigma_i\,\tau_i \;\in[-1,1]
\ee
and the covariance matrix of the Gaussian process $\HSKN$ is:
\be\label{eq:cov HSK}
\E \,\HSKN(\sigma)\, \HSKN(\tau) \,=\, N\, q_N(\sigma,\tau)^2 \;.
\ee
Given an inverse temperature $\beta>0$, the random partition function of the SK model is
\be
\ZSKN \,\equiv\, \sum_{\sigma\in\{-1,1\}^N} \exp\Bigg( -\beta\,\HSKN(\sigma) \,+\, \sum_{i=1}^N \tilde h_i\,\sigma_i\Bigg)
\ee
where $\tilde h_i$, $i=1,\dots,N$ is a family of i.i.d. copies of a random variable $h$ such that $\E|h|<\infty\,$.
The quenched pressure density of the SK model is
\be\label{qpresSK}
\pSKN(\beta,h) \,\equiv\, \frac{1}{N}\,\E\,\log \ZSKN
\ee
where $\E$ to denote the expectation over all couplings $\tilde J_{ij}$'s and fields $\tilde h_i$'s.
The quenched pressure converges as $N\to\infty$ and many properties of its limit, that we will denote by $\pSK(\beta,h)\,$, have been investigated in the literature \cite{MPV,GT,Guerra,Tala,Panchenko-Book,Auchen1}.


\begin{theorem}\label{maint}
The quenched pressure of the DBM satisfies the following lower bound:
\be\label{main2}
\liminf_{N\to\infty}\pDBMN \,\geq
\sup_{\va\in\Rpz^{K-1}} \PDBM(\va) \;,
\ee
where, for every $\va = (a_p)_{p=1,\dots,K-1}\in\R_+^{K-1}$, the functional $\PDBM(\va)=\PDBM(\va;\,\beta,\lambda,h)$ is defined as:
\be\label{variationalp}
\PDBM(\va) \,\equiv\, 
\sum_{p=1}^K \lambda_p\, \pSK\!\big(\theta_p(\va),h^{(p)}\big) \,-\,
\frac{1}{2}\,\sum_{p=1}^K \lambda_p\,\theta_p(\va)^2 \,+ \sum_{p=1}^{K-1}\lambda_p\,\betap^2\,\lambda_{p+1}
\ee
and the parameter $\theta_p(a)=\theta_p(a;\beta,\lambda)\geq0$ is defined by:
\be\label{gammap}
\theta_p(a)^{\,2} \,\equiv\, \begin{cases}
\,\lambda_1\,a_1\,\betauno^2\ & \textrm{for }p=1 \\[5pt]
\,\lambda_p\left(\,\dfrac{1}{a_{p-1}}\,\betapm^2 +\, a_p\,\betap^2\right)\ & \textrm{for }p=2,\dots,K-1 \\[8pt]
\,\lambda_K\,\dfrac{1}{a_{K-1}}\,\betaKm^2\ & \textrm{for }p=K \\[5pt]
\end{cases}\;.
\ee
\end{theorem}

\begin{proof}
We are going to prove the following lower bound at finite volume:
\be\label{main}
\pDBMN \,\geq\, \sum_{p=1}^K \lambda_p^{(N)}\, \pSKNp\big(\theta_p^{(N)},h^{(p)}\big) \,-\, \frac{1}{2}\,\sum_{p=1}^K \lambda_p^{(N)} \big(\theta_p^{(N)}\big)^2 \,+\, \sum_{p=1}^{K-1}\lambda_p^{(N)}\betap
^2\,\lambda_{p+1}^{(N)}
\ee
where $\theta_p^{(N)}\equiv \theta_p(a;\beta
,\lambda^{(N)})$ and $\va\in\Rpz^{K-1}$ can be arbitrarily chosen.
The lower bound \eqref{main2} will follow immediately by letting $N\to\infty$, since $\pSKN(\beta,h)$ is convex with respect to $\beta$ and thus the convergence to $\pSK$ is uniform on compact sets.

For every $p=1,\ldots,K$ let $\HSKp(s)$, $s\in\{-1,1\}^{L_p}\,$ be a Gaussian process representing the Hamiltonian of an SK model over the $N_p$ spin variables in the layer $L_p\,$.
We assume that $\HSKu,\dots,\HSKK$ are independent processes, also independent of the Hamiltonian $H_{\Lambda_N}$. 
For $\sigma\in\{-1,1\}^N$ and $t\in[0,1]$ we define an interpolating Hamiltonian as follows:
\be\label{inth}
\mathcal H_{N}(\sigma;t) \,\equiv\,
\sqrt{t}\; H_{\Lambda_N}(\sigma) \,+\,
\sqrt{1-t}\; \sum_{p=1}^K\, \theta_p^{(N)}\, \HSKp(\sigma_{L_p}) \;,
\ee
where of course $\sigma_{L_p}\equiv(\sigma_i)_{i\in L_p}\,$.
An interpolating quenched pressure is naturally defined as
\be\label{intpe}
\varphi_{N}(t)\,\equiv\, \frac{1}{N}\, \E\,\log\,\mathcal Z_{N}(t) \ ,
\ee
where
\be\label{partinn}
\mathcal Z_{N}(t)\,\equiv\, \sum_{\sigma\in\{-1,1\}^N} \exp\bigg(-\mathcal H_{N}(\sigma,t) \,+\, \sum_{p=1}^K\sum_{i\in L_p} h_i\,\sigma_i \bigg)
\ee
and $\E$ denotes the expectation with respect to all the couplings $J_{ij}$'s, $\tilde J_{ij}$'s, $h_i$'s.
The quenched pressure of the DBM and a convex combination of quenched pressures of SK models are recovered for $t=1$ and $t=0$ respectively:
\begin{align} \label{t1}
& \varphi_{N}(1) \,=\, \pDBMN \;,\\
\label{t0}
& \varphi_{N}(0) \,=\, \sum_{p=1}^K \lambda_p^{(N)}\, \pSKNp(\theta_p^{(N)},h^{(p)}) \;.
\end{align}
For every function $f:\{-1,1\}^{N}\times\{-1,1\}^{N}\rightarrow \R\,$ we denote
\be\label{qexp}
\left\langle\, f\,\right\rangle_{N,t} \,\equiv\,
\E\,\sum_{\sigma,\tau}\frac{e^{-\beta\,\mathcal H_{N}(\sigma;t)\,-\,\beta\,\mathcal H_{N}(\tau;t) \,+\, \sum_{p=1}^K\sum_{i\in L_p}h_i(\sigma_i+\tau_i)}}{\mathcal Z_N^2(t)}\,f(\sigma,\tau) \;.
\ee
Let $Q_N:\{-1,1\}^{N}\times\{-1,1\}^{N}\rightarrow \R\,$, 
\be \label{Qn}
Q_N\,\equiv\;
2\sum_{p=1}^{K-1} \lambda_p^{(N)}\betap
\,\lambda_{p+1}^{(N)}\; \qp\,\qpp \,-\,
\sum_{p=1}^K \lambda_p^{(N)} \big(\theta_p^{(N)}\big)^2 \qp^2 \;.
\ee
Gaussian integration by parts leads to the following result:
\be\label{deriv_ann}
\frac{d\varphi_N}{dt} \,=\,
\frac{1}{2}\, \Bigg(2\sum_{p=1}^{K-1}\lambda_p^{(N)}\,\betap
\,\lambda_{p+1}^{(N)} \,-\, \sum_{p=1}^K \lambda_p^{(N)} \big(\theta_p^{(N)}\big)^2 \Bigg) \,-\,
\frac{1}{2}\,\Big\langle Q_N \Big\rangle_{N,t} \;.
\ee
Now replacing the definition \eqref{gammap} of $\theta_p^{(N)}=\theta_p(a;\beta,\lambda^{(N)})$ into \eqref{Qn}, we obtain
\be \label{signQn}
Q_N \,=\, -\sum_{p=1}^{K-1} \betap^2\; \bigg(\lambda_{p+1}^{(N)}\,\frac{1}{\sqrt{a_p}}\;\qpp \,-\, \lambda_p^{(N)}\,\sqrt{a_p}\;\qp \bigg)^{\!2} \,\leq 0 \;.
\ee
The claim \eqref{main} follows immediately from \eqref{t1}, \eqref{t0}, \eqref{deriv_ann} and \eqref{signQn}.
\end{proof}

\begin{remark}\label{rk:DBMstationary}
$a=(a_p)_{p=1,\dots,K-1}$ is a stationary point of $\PDBM$ if and only if 
\be \label{ceDBM}
\frac{1}{a_p}\;\lambda_{p+1}\,\qSK\!\big(\theta_{p+1}(a),h^{(p+1)}\big) \,=\, \lambda_p\,\qSK\!\big(\theta_p(a),h^{(p)}\big)
\ee
for every $p=1,\dots,K-1\,$, where we define $\qSK(\beta,h)\geq0$ by
\be 
\qSK(\beta,h)^2\,\equiv\, \lim_{N\to\infty}\,\E\sum_{\sigma,\tau\in\{-1,1\}^N} q_N(\sigma,\tau)^2\; \mu_N^\textup{SK}(\sigma,\tau)\\[-15pt]
\ee
and
\be 
\mu_N^\textup{SK}(\sigma,\tau) \,\equiv\, \frac{1}{\big(\ZSKN\big)^2}\,\exp\bigg(-\beta\HSKN(\sigma)\,-\,\beta\HSKN(\tau)\,+\,\sum_{p=1}^K\sum_{i\in L_p} h_i\,(\sigma_i+\tau_i)\,\bigg) \;.
\ee
Since $\frac{\partial}{\partial\beta}\pSK(\beta,h)=\beta\,\big(1-\qSK(\beta,h)^2\big)\,$ \cite{Tala}, it is straightforward to compute $\frac{\partial}{\partial a_p}\PDBM\,$ from definition \eqref{variationalp} and find the stationary condition \eqref{ceDBM}.
\end{remark}

\section{The annealed region of the DBM} \label{sec:ann}

In this Section we consider the model in absence of external field ($h=0$) and we identify a region where the quenched and the annealed pressure of the DBM coincide.

\begin{definition}
The annealed pressure of the DBM is
\be 
\pADBM \,\equiv\, \lim_{N\to\infty} \frac{1}{N}\log\E\,Z_{\Lambda_N} \;.
\ee 
\end{definition}
It can be easily computed due to the Gaussian nature of the model:
\be \label{annpressureDBM}
\pADBM(\beta,\lambda) \,=\, \log 2 \,+\, \sum_{p=1}^{K-1}\lambda_p\,\betap^2\,\lambda_{p+1} \,.
\ee
By concavity of the $\log$, the annealed pressure is an upper bound for the quenched one:
\be \label{JensenDBM}
\limsup_{N\to\infty} \pDBMN \,\leq\, \pADBM \;.
\ee
The system is said to be in the \textit{annealed regime} when the parameters $(\beta,\lambda)$ are such that $\lim_{N\to\infty}\pDBMN = \pADBM\,$.

By Theorem \ref{maint} we can investigate the annealed regime of the DBM relying on the established results for the annealed regime of the SK model.
Let $\pSK$ be the limiting quenched pressure of an SK model and let $\pASK\equiv\lim_{N\to\infty}N^{-1}\log\E \ZSKN$ be its annealed version. Clearly:
\be\label{anndis}
\pSK \,\leq\, \pASK \,=\, \log 2 +\frac{\beta^2}{2} \;.
\ee
Equality is achieved in the so called annealed region of the SK model \cite{ALR,CG,Panchenko-Book,Tala}:
\be\label{anneq}
\pSK(\beta) \,=\,  \pASK(\beta) \quad \textrm{if }\beta^2\leq\frac{1}{2} \;.
\ee
Now consider the following system of inequalities:
\be\label{AK0}
\begin{cases}
\lambda_1\,a_1\,\betauno^2 <\, \dfrac{1}{2}\\[5pt]
\lambda_p\,\Big(\,\dfrac{1}{a_{p-1}}\,\betapm^2 +\, a_p\,\betap^2\Big) <\, \dfrac{1}{2}\ & \textrm{for }p=2,\dots,K-1 \\[8pt]
\lambda_K\,\dfrac{1}{a_{K-1}}\,\betaKm^2 <\, \dfrac{1}{2}
\end{cases}
\ee
and the following region of parameters of the DBM:
\be\label{AK}
A_K \equiv \Big\{(\beta,\vlambda)\in\Rpz^{K-1}\times T_K\ \Big|\ \exists\,\va\in\Rpz^{K-1} :
\eqref{AK0}\text{ is verified}
\Big\}\;,
\ee
where 
$T_K \equiv \{ (\lambda_1,\dots,\lambda_K) \in [0,1]^K \,|\, \sum_{p=1}^K\lambda_p=1 \}\,$ denotes the $K-$dimensional simplex.
We denote by $\overline{A_K}$ the topological closure of $A_K\,$.

\begin{theorem}\label{annt}
If $(\beta,\vlambda)\in \overline{A_K}$ there exists
\be\label{annDBM}
\lim_{N\to\infty} \pDBMN \,=\, \pADBM \,.
\ee
\end{theorem}

\begin{proof}
The lower bound \eqref{main2} for the quenched pressure of the DBM rewrites as:
\be\label{main 3}
\liminf_{N\to\infty}\pDBMN \,\geq\, \sup_{\va\in\Rpz^{K-1}}\,
\sum_{p=1}^K\, \lambda_p\, \Big(\pSK\left(\theta_p(\va)\right) \,-\, \pASK\left(\theta_p(\va)\right) \Big) \,+\, \pADBM \;.
\ee
Thanks to \eqref{anndis} and \eqref{anneq}, if $(\beta,\vlambda)\in \overline{A_K}$ then the supremum in \eqref{main 3} vanishes and
\be\label{queann}
\liminf_{N\to\infty}\pDBMN \,\geq\, \pADBM \;.
\ee
This bound together with \eqref{JensenDBM} concludes the proof.
\end{proof}

It is an open question whether $\overline{A_K}$ is the full annealed region of the system. We will see that Proposition \ref{stabilityprop} suggests a positive answer. We are now interested in a more explicit characterization of $A_K$. We mention that such a characterization can be interesting for inference problems as suggested in \cite{WK2}.
It is convenient to introduce the following family of polynomials.

\begin{definition} \label{def:Liebpoly}
Let $x\in\C$ and $t=(\tp)_{p=1,\dots,K-1}\in[0,\infty)^{K-1}$. We define recursively
\be \label{eq:Liebpoly} \begin{cases}
\,\DDelta_{p+1}(x,t) \,\equiv\, x\,\DDelta_{p}(x,t) - \tp\,\DDelta_{p-1}(x,t)\quad \text{for }p=1,\dots,K-1 \\[5pt]
\,\DDelta_1(x,t) \,\equiv\, x \,,\ 
\DDelta_0(x,t) \,\equiv\, 1
\end{cases}\;. \ee
\end{definition}

These orthogonal polynomials have several characterizations and were studied by Heilmann and Lieb \cite{HL,HLprl}. Some relevant properties can be found in the Appendix \ref{sec:appendix}.

\begin{remark}
The polynomial $\DDelta_K(x,t)$ has an interesting combinatorial interpretation. 
Let's denote by $\mathscr L_K$ the linear graph of vertex set $\{1,\dots,K\}$ and edge set $\{(p,p+1)\,|\,p=1,\dots,K-1\}\,$.
A \textit{matching} on $\mathscr L_K$ is a subset of pairwise disjoint edges.
Then:
\be \label{eq:match}
\DDelta_K(x,t) \,=\, \sum_{d=0}^{K/2} (-1)^d\, x^{K-2d}\, f_{d,K}(t) \ ,
\ee
where:
\be 
f_{d,K}(t) \,\equiv\, \sum_{\substack{D\text{ matching on }\mathscr{L}_K\\|D|=d}}\, \prod_{(p,p+1)\in D}\!\!\!\!t_{p} \ .
\ee
Indeed the polynomial on the right hand side of \eqref{eq:match} verifies the recursion relation \eqref{eq:Liebpoly} (see \cite{HL}).
\end{remark}

\begin{proposition} \label{annp}
Let $(\beta,\lambda)\in\Rpz^{K-1}\times T_K$ and consider the vector $t=(t_p)_{p=1,\dots,K-1}\,$ with
\be \label{eq:t}
\tp(\beta,\lambda) \,\equiv\, 4\,\lambda_p\,\betap^4\,\lambda_{p+1}
\ee
 for every $p=1,\dots,K-1$. Define
\be \label{rhoK}
\rho(\beta,\lambda) \,\equiv\, \max\big\{x>0 :\, \DDelta_K\!\big(x,\,t(\beta,\lambda)\big)=0 \big\} \;.
\ee
The followings are equivalent:
\begin{itemize}
\item[i)] $(\beta,\lambda)\in A_K$
\item[ii)] $\DDelta_p\!\big(1,\,t(\beta,\lambda)\big)>0\quad\forall\,p=2,\dots,K$
\item[iii)] $\rho(\beta,\lambda)<1$
\end{itemize}
\end{proposition}

\begin{proof}
i)$\Leftrightarrow$ii). To shorten the notation set $z_p\equiv\DDelta_p\!\big(1,\,t(\beta,\lambda)\big)$; by \eqref{eq:Liebpoly} we have
\be \label{eq:recursion zp} \begin{cases}
z_{p+1} \,=\, z_{p} \,-\, 4\,\lambda_p\,\betap^4\,\lambda_{p+1}\,z_{p-1} \quad \text{for }p=1,\dots,K-1 \\[5pt]
z_1 = 1 \,,\ 
z_0 = 1
\end{cases}\ .\ee
Set for every $p=1,\dots,K$
\be
a_p^*\,\equiv\,\frac{1}{2\,\lambda_p\,\betap^2}\,\frac{z_p}{z_{p-1}} \;;
\ee
then the following recursion relation follows from \eqref {eq:recursion zp}:
\be \begin{cases}
a_p^* \,=\, \dfrac{1}{2\,\lambda_p\,\betap^2} \,-\, \dfrac{\betapm^2}{\betap^2}\;\dfrac{1}{a_{p-1}^*} \quad \text{for }p=2,\dots,K \\[8pt]
a_1^* = \dfrac{1}{2\,\lambda_1\,\betauno^2}
\end{cases}\ .\ee
Now assume $z_1,\dots,z_K>0$. Then $a_1^*,\dots,a_K^*>0$ and choosing $a=(a_1^*,\dots,a_{K-1}^*)$ the system of inequalities \eqref{AK0} is verified.\\
On the other hand, assuming that there exists $a=(a_1,\dots,a_{K-1})\in\Rpz^{K-1}$ verifying \eqref{AK0}, one can prove by induction that $a_p^*\geq a_p>0$ for $p=1,\dots,K-1$ and $a_K^*>0\,$. Therefore $z_1,\dots,z_K>0\,$.

ii)$\Leftrightarrow$iii). The equivalence among these conditions is a consequence of the interlacing property of the zeros of $\DDelta_p\,$. A detailed proof can be found in the Appendix (Corollary \ref{cor:zeros} with $\rho=1$).
\end{proof}

\begin{remark} \label{rk:MK}
The polynomial $\DDelta_K(x,t)$ with $t=t(\beta,\lambda)$ defined in \eqref{eq:t} has also a linear algebra interpretation. Set:
\be \label{MK} \begin{split}
M(\beta,\lambda) \,&\equiv\, 2\,M_0(\beta)\,\diag(\lambda) \,=\\[5pt]
&=\, 2\,\begin{pmatrix}
0 					& \beta_1^2\lambda_2	& 	 					&								& \\[5pt]
\lambda_1\beta_1^2	& 0 					& \beta_2^2\lambda_3	& 								& \\[5pt]
					& \lambda_2\beta_2^2	& 0						&  								& \\
					& 						& 						& \ddots 						& \\
					&						&						&						   		& \beta_{K-1}^2\lambda_K \\
					&						&						& \lambda_{K-1}\beta_{K-1}^2 	& 0 \\
\end{pmatrix}
\end{split}\ee
where $M_0(\beta)$ is defined by \eqref{M0}.
The characteristic polynomial of $M(\beta,\lambda)$ is actually
\be \label{charpoly}
\DDelta_K\!\big(x,\,t(\beta,\lambda)\big) \,=\, \det\big( x\,I - M(\beta,\lambda) \big) \;.
\ee
Indeed using the Laplace expansion according to the last line of the matrix, it is easy to verify that the determinant on the right hand side of \eqref{charpoly} satisfies the recursion relation \eqref{eq:Liebpoly}.
Now since the zeros of $x\mapsto\DDelta_K(x,t(\beta,\lambda))$ are all real and symmetric with respect to the origin (see the Appendix), 
the largest one is the spectral radius of $M(\beta,\lambda)\,$:
\be \label{spectralradius}
\rho(\beta,\lambda) \,=\, \max\{|x| : x\text{ eigenvalue of }M(\beta,\lambda) \} \;.
\ee
\end{remark}

The next Proposition exploits the result of Proposition \ref{annp} in order to study the role of the parameters $\beta$ and $\lambda$ in the annealed behaviour of the system.

%

\begin{proposition} \label{prop:rhobound}
i) For every $\beta\in\R_+^{K-1}\,$, 
\be  \label{rhobound}
\sup_{\lambda\in T_K}\rho(\beta,\lambda) \,=\, \max_{p=1,\dots,K-1} \beta_p^2 \ .
\ee
The supremum is reached exactly for those $\lambda=\lambda^*(\beta)\in T_K$ such that there exists $p^*\in\{1,\dots,K-1\}\,$:
\be \label{bestlambda1}
\lambda_{p^*} \,=\, \lambda_{p^*\!+1} \,=\, \frac{1}{2}
\quad,\quad \beta_{p^*}=\max_{p=1,\dots,K-1}\beta_p
\ee
or $p^*\in\{2,\dots,K-1\}\,$:
\be \label{bestlambda2}
\lambda_{p^*} \,=\, \lambda_{p^*\!-1}+\lambda_{p^*\!+1} \,=\, \frac{1}{2}
\quad,\quad \beta_{p^*}=\beta_{p^*\!-1}=\max_{p=1,\dots,K-1}\beta_p \;.
\ee
ii) Moreover for every $\lambda\in T_K$, $\rho(\beta,\lambda)$ is a non-decreasing function of each $\beta_p$ for $p=1,\dots,K-1$.
\end{proposition}

Physically \textit{ii)} means that increasing the local temperatures pushes the system toward the annealed region.
On the other hand \textit{i)} implies that if all the inverse temperatures $\beta_p<1$ for $p=1,\dots,K-1$, then the system is in the annealed regime for every choice of the form factors $\lambda$. Furthermore if this is not the case, the system can be driven out of the region $A_K$ by localizing the positive density layers around the minimal temperature(s).

In order to prove Proposition \ref{prop:rhobound} we need the following elementary (but useful)
\begin{lemma} \label{lem:diseq}
Let $P\geq2$, $x_1,\dots,x_P\geq0$ and $b_1,\dots,b_{P-1}\geq0\,$. Set $S\equiv \sum_{p=1}^Px_p$ and $B \equiv \max_{p=1,\dots,P-1} b_p\,$. Then:
\be \label{diseq}
4\,\sum_{p=1}^{P-1} b_p\,x_p\,x_{p+1} \,\leq\, B \,S^2 \;.
\ee
Moreover we have equality in \eqref{diseq} if and only if there exists $p^*\in\{2,\dots,P-1\}$ such that
\be \label{onehalf1}
x_{p^*} \,=\, x_{p^*-1}+x_{p^*+1} \,=\, \frac{S}{2} \quad,\quad b_{p^*-1} = b_{p^*} = B
\ee
or there exists $p^*\in\{1,\dots,P-1\}$ such that
\be \label{onehalf2}
x_{p^*} \,=\, x_{p^*+1} \,=\, \frac{S}{2} \quad,\quad b_{p^*} = B \;.
\ee
\end{lemma}

\begin{proof}
Since
\be 
0\,\leq\,
\bigg(\sum_p (-1)^p\, x_p\bigg)^2 =\,
\sum_p x_p^2 \,+\, 2\,\sum_{p<p'}(-1)^{p+p'}x_p\,x_{p'} \;,
\ee
the following inequality holds true:
\be 
\sum_p x_p^2 \;\geq\, -2\,\sum_{p<p'}(-1)^{p+p'}x_p\,x_{p'} \;.
\ee
Therefore:
\be \begin{split}
\bigg(\sum_p x_p\bigg)^2 &=\,
\sum_p x_p^2 \,+\, 2\sum_{p<p'}x_p\,x_{p'} \\[6pt]
&\geq\, 2\,\sum_{p<p'}\Big(1-(-1)^{p+p'}\Big)\,x_p\,x_{p'} \\[6pt]
&\geq\, 4\,\sum_{p}x_p\,x_{p+1} \;.
\end{split}\ee
As a trivial consequence we have:
\be 
4\,\sum_{p}b_p\,x_p\,x_{p+1} \,\leq\, 4\,B\sum_{p}x_p\,x_{p+1} \,\leq\, B\,\bigg(\sum_p x_p\bigg)^2\;.
\ee
Now all the previous inequalities are saturated if and only if the following conditions are fulfilled:
\be \label{onehalf0} \begin{cases}
\,\sum_{p\,\text{even}}x_p=\sum_{p\,\text{odd}}x_p \\[8pt]
\, x_p\,x_{p'}=0 \quad\forall\,p,p':\,p+p'\,\text{odd},\, p\leq p'+3 \\[6pt]
\, b_p=B \quad\forall\,p:\,x_p\,x_{p+1}\neq0
\end{cases} \;.\ee
It is easy to check that \eqref{onehalf0} is equivalent to \eqref{onehalf1} or \eqref{onehalf2}, concluding the proof.
\end{proof}

\begin{proof}[Proof of Proposition \ref{prop:rhobound}]
By Remark \ref{rk:MK}, $\rho(\beta,\lambda)$ is the spectral radius of the matrix $M(\beta,\lambda)\,$. Hence:
\be 
\rho(\beta,\lambda) \,\leq\, \Vert M(\beta,\lambda)^2 \Vert_\infty^{1/2}
\ee
and the square of the matrix \eqref{MK} can be easily computed leading to
\be \label{norma2} \begin{split}
\Vert M(\beta,\lambda)^2 \Vert_\infty \,&=\,
4\,\max_{p=1,\dots,K} \sum_{p'=p-2}^{p+1} b_{p'}^{(p)}\, \lambda_{p'}\,\lambda_{p'+1} \\
&\leq\, \max_{p=1,\dots,K-1} \beta_p^4 \;,
\end{split}\ee
where for every $p=1,\dots,K$, $p'=p-2,\dots,p+1$ we set
\be 
b_{p'}^{(p)} \,\equiv\, \beta_{p-2}^2\,\beta_{p-1}^2\;\delta_{p-2,p'} +\,
\beta_{p-1}^4\;\delta_{p-1,p'} +\,
\beta_{p}^4\;\delta_{p,p'} +\,
\beta_{p}^2\,\beta_{p+1}^2\;\delta_{p+1,p'}
\ee
and for convenience we denote $\lambda_{p}\equiv0$ for $p\notin\{1,\dots,K\}\,$ and $\beta_{p}\equiv0$ for $p\notin\{1,\dots,K-1\}$.
The inequality in \eqref{norma2} follows by Lemma \ref{lem:diseq} since $\sum_p \lambda_p=1\,$.

Now assume that $\rho(\beta,\lambda)=\max_{p=1,\dots,K-1}\beta_p^2\equiv\hat\beta^2$. In particular the inequality in \eqref{norma2} must be saturated, namely there exists $p\in\{1,\dots,K\}$ such that
\be 
4\sum_{p'=p-2}^{p+1}b_{p'}^{(p)}\,\lambda_{p'}\,\lambda_{p'+1} \,=\, \hat\beta^4 \;.
\ee
Then \eqref{bestlambda1} or \eqref{bestlambda2} follow from Lemma \ref{lem:diseq}.

On the other hand assume that condition \eqref{bestlambda1} or \eqref{bestlambda2} holds true. In order to prove that $\rho(\beta,\lambda)=\hat\beta^2$, it suffices to show that $x=\hat\beta^2$ is a zero of the matching polynomial $\DDelta_K\!\big(x,t(\beta,\lambda)\big)$, where the activities vector $t(\beta,\lambda)$ is defined by \eqref{eq:t}.
Now condition \eqref{bestlambda2} implies that
\be 
\DDelta_K\!\big(\hat\beta^2,\,t(\beta,\lambda)\big) \,=\, \hat\beta^{2K}\,\big(1-4\,\lambda_{p^*}\lambda_{p^*\!+1}\big) \,=\, 0 \;;
\ee
while condition \eqref{bestlambda1} implies that
\be 
\DDelta_K\!\big(\hat\beta^2,\,t(\beta,\lambda)\big) \,=\, \hat\beta^{2K}\,\big(1-4\,\lambda_{p^*\!-1}\lambda_{p^*}-4\,\lambda_{p^*}\lambda_{p^*\!+1}\big) \,=\, 0 \;.
\ee
This concludes the proof of Proposition \ref{prop:rhobound} part \textit{i)}. 
In order to prove part \textit{ii)}, we observe that the matrix $M(\beta,\lambda)$ has non-negative entries, therefore its spectral radius $\rho(\beta,\lambda)$ is a non-decreasing function of its entries.
\end{proof}

\section{The replica symmetric ansatz for the DBM} \label{Sec:RS}

In this section we derive a replica symmetric expression for the pressure of the DBM.
We show that at zero magnetic field the annealed region $A_K$ identified by Theorem \ref{annt} and Proposition \ref{annp} is the only region where the annealed solution is stable for the replica symmetric consistency equation.
Finally we prove the uniqueness of the solution of the replica symmetric consistency equation, under the hypothesis of Gaussian centred external fields.



Let $q=(q_p)_{p=1,\dots,K}\in[0,1]^K\,$. Consider the matrices $M=M(\beta,\lambda)$, $M_1=M_1(\beta,\lambda)$ defined by \eqref{MK}, \eqref{M1} respectively.
For $p=1,\dots,K$ we have
\be \label{Mq}
\big(M q\big)_p \,=\,
2\,q_{p-1}\,\lambda_{p-1}\,\betapm^2 +\, 2\,\betap^2\,\lambda_{p+1}\,q_{p+1}
\ee
where $\beta_0=\beta_K=\lambda_0=\lambda_{K+1}=q_0=q_{K+1}\equiv0\,$ for convenience. We have
\be
\frac{1}{2}\,q^T M_1\,q \,=\, \sum_{p=1}^{K-1}\lambda_p\,\beta_p^2\,\lambda_{p+1}\, q_p\,\,q_{p+1} \;.
\ee


\begin{definition} \label{def:rs}
For every $q=(q_p)_{p=1,\dots,K}\in[0,1]^K$ the \textit{replica symmetric functional} of the DBM is
\be\label{PRS}
\begin{split}
\pRSN(q) \,\equiv\  & 
\sum_{p=1}^{K} \lambda_p^{(N)}\; \E\log\cosh\left(z\,\sqrt{\big(M^{(N)}\,q\big)_p\,}+h^{(p)}\right) \,+\\
& +\, \frac{1}{2}\,(1-q)^T\,M_1^{(N)}\,(1-q) \,+\, \log 2 
\end{split}
\ee
where $z$ is a standard Gaussian random variable independent of $h$ and $M^{(N)}\equiv M(\beta
,\lambda^{(N)})\,$, $M_1^{(N)}\equiv M_1(\beta
,\lambda^{(N)})$ are tridiagonal matrices defined by \eqref{MK}, \eqref{M1} respectively. 
The limit of the functional as $N\to\infty$ is
\be\label{PRSlimit}
\begin{split}
\pRS(q;\,\beta,\lambda,h) \,\equiv\ & 
\sum_{p=1}^{K} \lambda_p\; \E\log\cosh\left(z\,\sqrt{\big(M q\big)_p\,}+h^{(p)}\right) \,+\\
& +\, \frac{1}{2}\,(1-q)^T\,M_1\,(1-q) \,+\, \log 2 
\end{split}
\ee
where $M=M(\beta,\lambda)$ and $M_1=M_1(\beta,\lambda)\,$.
\end{definition}

Definition \ref{def:rs} is motivated by the following

\begin{proposition}
For every $q=(q_p)_{p=1,\dots,K}\in[0,1]^K$
\be\label{ftc}
\pDBMN \,=\, \pRSN(q) \,-\,
\frac{1}{2}\, \int_0^1 \Big\langle \big(q_{\Lambda_N}-q\big)^T M_1^{(N)} \big(q_{\Lambda_N}-q\big) \Big\rangle_{N,t} \,dt
\ee
where $q_{\Lambda_N}\equiv\big(q_{L_p}(\sigma,\tau)\big)_{p=1,\dots,K}$ and $\langle\,\cdot\,\rangle_{N,t}$ denotes the quenched Gibbs expectation associated to a suitable Hamiltonian. 
\end{proposition}

\begin{proof}
Let $q\in[0,\infty)^K$. For every $p=1,\dots, K$ we consider a one-body model over the $N_p$ spin variables indexed by the layer $L_p$ at inverse temperature $\sqrt{(M^{(N)}q)_p}\,$ and external fields distributed as $h^{(p)}$.
For $\sigma\in\{-1,1\}^N$ and $t\in[0,1]$ we define an interpolating Hamiltonian as follows:
\be \label{Hinter_rs}
\mathcal H_N(\sigma,t) \,\equiv\,
\sqrt{t}\; H_{\Lambda_N}(\sigma) \,+\,
\sum_{p=1}^K \,\sum_{i\in L_p} \left(z_i\, \sqrt{(1-t)\,(M^{(N)}q)_p}\,+h_i\right)\sigma_i
\ee
where $z_i$, $i\in L_p$, $p=1,\dots, K$ are i.i.d. standard Gaussian random variables, independent also of $h_i$'s and $J_{ij}$'s.
The interpolating pressure is
\be
\varphi_N(t) \,\equiv\, \frac{1}{N}\,\E\,\log\,\sum_{\sigma} e^{-\mathcal H_{N}(\sigma,t)} \ .
\ee
Observe that the quenched pressure of the DBM and a convex combination of quenched pressures of one-body models are recovered for $t=1$, $t=0$ respectively:
\begin{align} \label{t1_rs}
& \varphi_{N}(1) \,=\, \pDBMN \;,\\
\label{t0_rs}
& \varphi_{N}(0) \,=\,
\log 2 \,+\, \sum_{p=1}^K \lambda_p^{(N)}\; \E\log\cosh\left(z\,\sqrt{(M^{(N)}q)_p\,}+h^{(p)}\right) \;.
\end{align}
Gaussian integration by parts leads to the following result:
\be\label{deriv_rs}
\frac{d\varphi_N}{dt}\,(t) \,=\,
\frac{1}{2}\,(1-q)^T M_1^{(N)} (1-q) \,-\, \frac{1}{2}\,\Big\langle \big(q_{\Lambda_N}-q\big)^T M_1^{(N)} \big(q_{\Lambda_N}-q\big) \Big\rangle_{N,t}
\ee
where 
$\langle\,\cdot\,\rangle_{N,t}$ denotes the quenched Gibbs expectation associated to the Hamiltonian $\mathcal H_N(\sigma,t)+\mathcal H_N(\tau,t)$.
Therefore \eqref{ftc} follows by \eqref{t1_rs}, \eqref{t0_rs}, \eqref{deriv_rs} concluding the proof.
\end{proof}

We say that the DBM is in the \textit{replica symmetric regime} when there exists $q^*$ stationary point of $\pRS(q)$ such that $\lim_{N\to\infty}\pDBMN = \pRS(q^*)\,$. 

\begin{remark} \label{rk:RSstationary}
$q=(q_p)_{p=1,\dots,K}$ is a stationary point of $\pRS$ if and only if
\be \label{RScons}
M_1 \cdot \left(q_p-\E\tanh^2\left(z\,\sqrt{(Mq)_p\,}+h^{(p)}\right)\right)_{p=1,\dots,K} \,=\, 0
\ee
where the matrices $M=M(\beta,\lambda)$, $M_1=M_1(\beta,\lambda)$ are defined by \eqref{MK}, \eqref{M1} respectively and $z$ is a standard Gaussian random variable independent of $h$.
Indeed Gaussian integration by parts allows to compute $\frac{\partial}{\partial q_p}\pRS$ from definition \eqref{PRSlimit}.
\end{remark}

\begin{remark} 
For $h=0$ observe that $q=0$ is a solution of \eqref{RScons} and the replica symmetric functional computed at this stationary point equals the annealed pressure of the DBM:
\be
\pRS\big(q=0;\,\beta,\lambda,h=0\big) \,=\, \pADBM(\beta,\lambda) \;.
\ee
\end{remark}

\begin{proposition} \label{stabilityprop}
Set $F:[0,1]^K\to[0,1]^K$, $F_p(q) \,\equiv\, \E\tanh^2\!\left(z\,\sqrt{(Mq)_p}\right)$
for every $p=1,\dots,K$.
The region of parameters $(\beta,\lambda)$ such that the annealed solution $q=0$ is a stable solution of the replica symmetric consistency equation $q=F(q)$
coincides with the region $A_K$ introduced in Section \ref{sec:ann}. Precisely:
\be\label{stability}
|x|<1\ \;\forall\,x\,\text{eigenvalue of }\Jac F\Big|_{q=0} \quad \Leftrightarrow\quad (\beta,\lambda)\in A_K\;.
\ee
\end{proposition}

\begin{proof}
Gaussian integration by parts allows to compute the derivatives of $F$ with respect to $q$, leading to
\be
\Jac F\,\Big|_{q=0} \,=\, M \;.
\ee
Therefore \eqref{stability} follows immediately by Proposition \ref{annp} and Remark \ref{rk:MK}.
\end{proof}

When the matrix $M_1$ is invertible, the replica symmetric equation \eqref{RScons} rewrites as:
\be \label{RScons2}
q_p \,=\, \E\tanh^2\left(z\,\sqrt{(Mq)_p\,}+h^{(p)}\right) \quad\forall\ p=1,\dots,K \;.
\ee
The problem of uniqueness of the solution of \eqref{RScons2} has been proposed by Panchenko in \cite{PanchenkoMSK} for the convex case (where $M$ is replaced by a positive definite matrix) and solved in \cite{BSS} for $K=2$. In the following we prove the uniqueness for the deep case (our matrix $M$ is highly non-definite) under the assumption of Gaussian centred external fields.
Denote $T_K^+ \,\equiv\, \{ (\lambda_1,\dots,\lambda_K) \in (0,1]^K \,|\, \sum_{p=1}^K\lambda_p=1 \}\,$.

\begin{theorem} \label{teo:unique}
Let $h^{(p)}$, $p=1,\dots,K$ be centered Gaussian variables with variance $v_p>0$ respectively.
Let $\lambda\in T_K^+$ and $\beta\in\R_+^{K-1}$.
The consistency equation \eqref{RScons2}, which rewrites as
\be \label{RSconsGauss}
q_p\,=\,\E \tanh^2\left(z\,\sqrt{(M q)_p+v_p}\,\right) \qquad \forall\,p=1,\dots,K
\ee
with $M=M(\beta,\lambda)$ defined in \eqref{MK}, has a unique solution.
\end{theorem}

The proof of Theorem \ref{teo:unique} relies on the following

\begin{lemma} \label{lem:monotoniaQ}
Let $h$ be a centered Gaussian variable with variance $v>0$. Let $\beta>0$. Then equation
\be \label{ceSK-RSGauss}
q\,=\,\E \tanh^2\left(z\,\sqrt{2\,q\,\beta^2+v}\,\right)
\ee
has a unique solution that we denote by $\qSKRS(\beta,v)>0\,$.
The function $\qSKRS$ is strictly increasing with respect to both $\beta$ and $v$.
\end{lemma}

The uniqueness part in Lemma \ref{lem:monotoniaQ} is the well-known Latala-Guerra's lemma \cite{Tala}. The monotonicity part is based on a similar argument.
Whereas the uniqueness property holds true for much more general choices of the external field $h$, we notice that the monotonicity property in $\beta$ is lost for deterministic (large enough) $h$.

\begin{proof}[Lemma \ref{lem:monotoniaQ}]
\ Set $f(q)\equiv q^{-1}\,\E \tanh^2(z\,\sqrt{2\,q\,\beta^2+v}\,)$ for $q>0$. To prove that \eqref{ceSK-RSGauss} has a unique solution it suffices to show that $f$ is strictly decreasing. Now taking the derivative of $f$ (avoiding Gaussian integration by parts) leads to:
\be \label{appo}
q^2\,\frac{df}{d q} \,=\,
-\,\E\left[\phi(y)\,\left(\phi(y)- y\,\phi'(y)\,\right)\right]
-\, \frac{v}{2\,q\,\beta^2+v}\;\E\left[y\,\phi(y)\,\phi'(y)\right]
\ee
where $\phi(y)\equiv\tanh y$ and $y\equiv z\,\sqrt{2\,q\,\beta^2+v}\,$.
Since $\phi$ is odd, strictly positive on $\R_+$, strictly increasing on $\R$ and strictly concave on $\R_+$, it follows that the functions inside each expectation in \eqref{appo} are strictly positive for $y\neq0\,$. In particular observe that $\sign\phi(y)=\sign y$ and that
\be
\frac{d}{dy}\,\big(\phi(y)-y\,\phi'(y)\big) = -y\,\phi''(y) \,>\,0\ \ \Rightarrow\ \  \sign\big(\phi(y)-y\,\phi'(y)\big) = \sign y \;.
\ee
Therefore $\frac{df}{dq}<0$, proving uniqueness of the solution of equation \eqref{ceSK-RSGauss}.

Now let's prove that the solution $\qSKRS$ is strictly increasing with respect to $\beta>0$. Taking the derivative with respect to $\beta^2$ on both sides of \eqref{ceSK-RSGauss} (avoiding integration by parts), one finds:
\be \label{appo2}
\frac{d\qSKRS}{d\beta^2} \,=\, \frac{\E\big[Y\phi(Y)\,\phi'(Y)\big]}{2\,\beta^2\,\qSKRS+v}\; \left(2\,\beta^2\,\frac{d\qSKRS}{d\beta^2}+\,2\qSKRS\right)
\ee
where $Y\equiv z\,\sqrt{2\beta^2\qSKRS+v}\,$. Reordering terms and replacing $\qSKRS$ by $\E\,\phi(Y)^2$ leads to:
\be
\frac{d\qSKRS}{d\beta^2} \,=\, \frac{ \E\big[Y\phi(Y)\,\phi'(Y)\big]\;2\qSKRS}{v+2\,\beta^2\,\E\big[\phi(Y)\,\left(\phi(Y)- Y\,\phi'(Y)\,\right)\big] } \,>0\;.
\ee
In a similar way one can prove that $\qSKRS$ is strictly increasing with respect to $v$, indeed:
\be
\frac{d}{d v}\qSKRS \,=\, \frac{\E\big[Y\phi(Y)\,\phi'(Y)\big]}{v+2\,\beta^2\,\E\big[\phi(Y)\,\left(\phi(Y)- Y\,\phi'(Y)\,\right)\big] } \,>0\;.
\ee
\end{proof}

\begin{proof}[Theorem \ref{teo:unique}]
A key observation is that the system \eqref{RSconsGauss} is equivalent to the following:
\be \label{RSconsq-a}
\begin{cases}
\,q_p \,=\, \E\tanh^{2}\left(z\,\sqrt{2\,q_p\,\theta_p(a)^2+v_p}\,\right) & p=1,\dots,K\\[8pt]
\,\lambda_p\,q_p\;a_p \,=\, \lambda_{p+1}\,q_{p+1} & p=1,\dots,K-1
\end{cases}
\ee
where we have introduced the auxiliary variables $a_1,\dots,a_{K-1}>0\,$.
This can be easily checked by comparing definitions \eqref{gammap} and \eqref{Mq}.
By Lemma \ref{lem:monotoniaQ}, the first line of \eqref{RSconsq-a} entails
\be
q_p \,=\, \qSKRS\!\big(\theta_p(a),v_p\big) \quad\forall\,p=1,\dots,K
\ee
where $\qSKRS$ is uniquely defined and strictly increasing with respect to both arguments.
On the other hand the second line of \eqref{RSconsq-a} rewrites as
\be
\lambda_1\,q_1\;\prod_{l=1}^p a_l \,=\, \lambda_{p+1}\,q_{p+1} \quad\forall\,p=1,\dots,K-1 \;.
\ee
Therefore in order to prove the Theorem it suffices to prove uniqueness of the solution $a\in\R_+^{K-1}$ of the following system:
\be \label{RSconsA}
\lambda_1\,\qSKRS\!\big(\theta_1(a),v_1\big)\,\prod_{l=1}^{p} a_l \,=\, \lambda_{p+1}\,\qSKRS\!\big(\theta_{p+1}(a),v_{p+1}\big) \quad\forall\,p=1,\dots,K-1 \,.
\ee
It is convenient to set $Q_1(a_1)\,\equiv\, \lambda_1\,\qSKRS\big(\lambda_1\,\beta_1^2\,a_1,v_1\big)$ and for every $p\geq2$
\be 
Q_p\bigg(\frac{1}{a_{p-1}}\,,a_{p}\bigg) \,\equiv\, 
\lambda_p\,\qSKRS\bigg(\lambda_p\,\frac{\beta_{p-1}^2}{a_{p-1}}+\lambda_p\,\beta_p^2\,a_p\,,\,v_p\bigg) \;.
\ee
We are going to prove by induction on $p\geq1$ that for any given $a_{p+1}\geq0$ there exists a unique $a_p\,=\,a^*_p(a_{p+1})>0$ such that
\be \label{starp}\begin{cases}
\;a_l = a_l^*(a_{l+1}) \quad\forall\;l=1,\dots,p-1\\[6pt]
\;Q_1(a_1)\;a_1\,\cdots\,a_{p-1}\,a_p \,=\, Q_{p+1}\bigg(\dfrac{1}{a_p}\,,a_{p+1}\bigg) 
\end{cases} \ee
and moreover $a_p^*$ is strictly increasing with respect to $a_{p+1}\,$. The uniqueness of solution of \eqref{RSconsA} will follow immediately by stopping at $p=K-1$ and choosing $a_K=0\,$.\\
$\bullet$ Case $p=1$: given $a_2\geq0$, let's consider the equation
\be \label{star1}
Q_1(a_1)\,a_1 \,=\, Q_2\bigg(\frac{1}{a_1},a_2\bigg) \;.
\ee
By Lemma \ref{lem:monotoniaQ} the left-hand side of \eqref{star1} is a strictly increasing function of $a_1>0$ and takes all the values in the interval $(0,\infty)$, while the right-hand side is a decreasing function of $a_1>0$ and takes non-negative values.
Therefore there exists a unique $a_1=a_1^*(a_2)>0$ solution of \eqref{star1}.
Now taking derivatives on both sides of \eqref{star1} and using again Lemma \ref{lem:monotoniaQ}, one finds:
\be 
\frac{d a_1^*}{d a_2} \,=\,
\frac{\partial}{\partial a_2}Q_2\Big(\frac{1}{a_1},a_2\Big) \, \Bigg[\frac{\partial}{\partial a_1}\big(Q_1(a_1)\,a_1\big) - \frac{\partial}{\partial a_1}Q_2\Big(\frac{1}{a_1},a_2\Big) \Bigg]^{-1}_{|a_1=a_1^*(a_2)} >0 
\ee
hence $a_1^*$ is a strictly increasing function of $a_2\,$.\\
$\bullet$ For $p>1\,$, $p-1$ $\Rightarrow$ $p$. Fix $a_{p+1}\geq0\,$.
By inductive hypothesis $a_1^*,\dots,a_{p-1}^*$ are well-defined and strictly increasing functions.
Defining the composition $A_l^*\equiv a_l^*\circ\dots\circ a_{p-1}^*$ for every $l=1,\dots,p-1$, equation \eqref{starp} rewrites as:
\be \label{starp2}
\big(Q_1\circ A_1^*\big)(a_p) \, \prod_{l=1}^{p-1}\!A_l^*(a_p)\; a_p \,=\, Q_{p+1}\bigg(\frac{1}{a_p},a_{p+1}\bigg) \;.
\ee
By inductive hypothesis and Lemma \ref{lem:monotoniaQ}, the left-hand side of \eqref{starp2} is a strictly increasing function of $a_p>0$ and takes all the values in the interval $(0,\infty)$, while the right hand-side of \eqref{starp2} is a decreasing function of $a_p>0$ and takes non-negative values. Therefore for every $a_{p+1}\geq0$ there exists a unique $a_p=a_p^*(a_{p+1})>0$ solution of \eqref{starp2}.
Now taking derivatives on both sides of \eqref{starp2} one finds:
\be \begin{split}
\frac{d a_p^*}{d a_{p+1}} = & \;
\frac{\partial}{\partial a_{p+1}}Q_{p+1}\Big(\frac{1}{a_p},a_{p+1}\Big)\; \cdot\\
&\cdot\Bigg[\frac{\partial}{\partial a_p}\bigg(\!\big(Q_1\circ A_1^*\big)(a_p) \prod_{l=1}^{p-1}A_l^*(a_p)\,a_p\bigg) - \frac{\partial}{\partial a_p}Q_{p+1}\Big(\frac{1}{a_p},a_{p+1}\Big) \Bigg]^{-1}_{|a_p=a_p^*(a_{p+1})}
\end{split} \ee
which, using again the inductive hypothesis and Lemma \ref{lem:monotoniaQ}, entails that $a_p^*$ is a strictly increasing function of $a_{p+1}\,$.
\end{proof}

\section{A replica symmetric bound for the DBM} \label{sec:rs bound}
In this section a lower bound for the quenched pressure of the DBM in terms of the replica symmetric functional is provided in a suitable region of the parameters $\beta,\,\lambda,\,h$.
For centred Gaussian external fields this region is defined though a system of $K$ inequalities which mimic the Almeida-Thouless condition for the SK model.

By Theorem \ref{maint} we can investigate the replica symmetric regime of the DBM relying on the established results for the replica symmetric regime of the SK model.
Denote by $\pSKRS$ the replica symmetric functional of an SK model, namely for every $q\in[0,1]$, $\beta>0$, $h$ real random variable with $\E\,|h|<\infty$,
\be\label{SK-RS}
\pSKRS(q;\,\beta,h) \,\equiv\, \E \log\cosh\left(z\,\sqrt{2\,q\,\beta^2}\,+h\right) \,+\, \frac{\beta^2}{2}\,(1-q)^2 \,+\, \log 2 
\ee
where $z$ is a standard Gaussian random variable independent of $h$.
Stationary points of $\pSKRS$ are identified by the consistency equation
\be \label{ceSK-RS}
q \,=\, \E\tanh^2\left(z\,\sqrt{2\,q\,\beta^2}\,+h\right)
\ee
where $z$ is a standard Gaussian r.v. independent of $h$.
%
The celebrated Guerra's bound \cite{Guerra} states in particular that
\be \label{boundSK-RS}
\pSK(\beta,h) \,\leq\, \inf_q \pSKRS(q;\,\beta,h) \;.
\ee
for every $\beta,h$. Identifying the exact replica symmetric region of the SK model, where equality in \eqref{boundSK-RS} is achieved, is an open problem.
A first result about the replica symmetric region of the DBM under general (but implicit) conditions is provided by the following

%

\begin{theorem} \label{teo:RS}
For every $q\in[0,1]^K$, $a\in\R_+^K$ related by
\be \label{ceA}
\lambda_p\,q_p\;a_p \,=\, \lambda_{p+1}\,q_{p+1} \quad\forall\,p=1,\dots,K-1
\ee
the following inequality holds true:
\be \label{disDBM-RS}
\PDBM(a;\,\beta,\lambda,h) \,\leq\, \pRS(q;\,\beta,\lambda,h) \;.
\ee
Moreover if the parameters $\beta,\,\lambda,\,h$ are such that there exist $q,\,a$ related by \eqref{ceA} and verifying
\be \label{hypSKRS}
\pSK\big(\theta_p(a),h^{(p)}\big)\,=\,\pSKRS\big(q_p\,;\theta_p(a),h^{(p)}\big) \quad\forall\,p=1,\dots,K \;,
\ee
then equality is achieved in \eqref{disDBM-RS}
and as a consequence
\be \label{boundDBM-RS}
\liminf_{N\to\infty} \pDBMN \,\geq\, \pRS(q;\,\beta,\lambda,h) \;.
\ee
\end{theorem}

\begin{proof}
Since $q,\,a$ are related by \eqref{ceA}, it is straightforward to verify that
\be \label{theta-M}
2\,q_p\,\theta_p(a)^2 \,=\, (Mq)_p \quad\forall\,p=1,\dots,K\,.
\ee
%
By Guerra's bound \eqref{boundSK-RS}, substituting $\pSKRS$ to $\pSK$ in the right-hand side of expression \eqref{variationalp} provides an upper bound to $\PDBM(a)\,$.
Now using the expression \eqref{SK-RS} of $\pSKRS$, the relation \eqref{theta-M} and comparing with the expression \eqref{PRSlimit} of $\pRS$, bound \eqref{disDBM-RS} is finally proved.\\
Following the same computations, if \eqref{hypSKRS} holds true then 
\be 
\PDBM(a;\,\beta,\lambda,h) \,=\, \pRS(q;\,\beta,\lambda,h)
\ee
and bound \eqref{boundDBM-RS} then follows by Theorem \ref{maint}.
\end{proof}

More explicit conditions for achieving equality in \eqref{disDBM-RS} and having the replica symmetric bound \eqref{boundDBM-RS} are based on the control of the replica symmetric region in the SK model.
For example it is known that equality in \eqref{boundSK-RS} is achieved for $\beta$ small enough. Precisely in Theorem 1.4.10 of \cite{Tala} Talagrand proves that for every $h$
\be\label{TalaSK}
\pSK(\beta,h) \,=\, \pSKRS(q;\,\beta,h) \quad\text{if } \beta^2<\frac{1}{8}
\ee
where $q$ is the unique solution of \eqref{ceSK-RS} (notice the different parametrisation with respect to \cite{Tala}).

\begin{corollary}
Let $\beta,\,\lambda,\,h$ such that a solution $q$ of the replica symmetric consistency equation \eqref{RScons2} satisfies the inequalities
\be \label{TalaDBM}
(Mq)_p < \frac{1}{4}\,q_p \quad\forall\,p=1,\dots,K
\ee
Then the replica symmetric bound \eqref{boundDBM-RS} holds true.
\end{corollary}

\begin{proof}
Let $q$ be a solution of \eqref{RScons2} satisfying \eqref{TalaDBM}. Let $a\in\R_+^{K-1}$ verifying \eqref{ceA}, so that the relation \eqref{theta-M} holds true. Then \eqref{TalaDBM} and \eqref{RScons2} rewrite respectively as:
\be \begin{cases}
\,\theta_p(a)^2 \,<\, \dfrac{1}{8} \\[6pt]
\,q_p \,=\, \E\tanh^{2}\left(z\,\sqrt{2\,q_p\,\theta_p(a)^2}+h^{(p)}\,\right)
\end{cases}\ee
for every $p=1,\dots,K\,$.
By Talagrand's result \eqref{TalaSK}, this entails
\be 
\pSK\big(\theta_p(a),h^{(p)}\big)\,=\,\pSKRS\big(q_p\,;\theta_p(a),h^{(p)}\big)
\ee
for every $p=1,\dots,K\,$.
Therefore by Theorem \ref{teo:RS}, 
\be 
\PDBM(a;\,\beta,\lambda,h) \,=\, \pRS(q;\,\beta,\lambda,h)
\ee
and the bound \eqref{boundDBM-RS} holds true.
\end{proof}

A complete characterization of the SK replica symmetric region where equality is achieved in \eqref{boundSK-RS} is still missing (see nevertheless \cite{Gut,Tala,JT}). A necessary condition is the Almeida-Thouless condition \cite{ton}:
\be\label{ATcond}
\beta^2\;\E\cosh^{-4}\left(z\,\sqrt{2\,q\,\beta^2}\,+h\right) \,\leq\, \frac{1}{2}
\ee
where $q$ is a solution of the consistency equation \eqref{ceSK-RS}.

However if we take $h$ Gaussian centered r.v. with variance $v>0$, it was recently proved \cite{Chen_private} that the Almeida-Thouless condition is also sufficient to have equality in \eqref{boundSK-RS}.
Precisely:
\be\label{ATGauss}
\pSK(\beta,h) = \pSKRS(q;\,\beta,h)
\ \Leftrightarrow\ 
\begin{cases}
\,\beta^2\;\E\cosh^{-4}\left(z\,\sqrt{2\,q\,\beta^2+v}\,\right) \,\leq\, \dfrac{1}{2} \\[8pt]
\,q \text{ is the (unique) solution of }\eqref{ceSK-RSGauss}
\end{cases}\;.
\ee

\begin{corollary}
Assume $h^{(p)}$, $p=1,\dots,K$ centered Gaussian variables of variance $v_p>0$ respectively.
Let $\beta,\,\lambda,\,v$ such that the (unique) solution $q$ of the replica symmetric consistency equation \eqref{RSconsGauss} satisfies the inequalities
\be \label{ATDBM}
(M q)_p\,\ \E \cosh^{-4}\left(z\,\sqrt{(M q)_p+v_p}\,\right) \,\leq\, q_p \qquad \forall\,p=1,\dots,K \;.
\ee
Then the replica symmetric bound \eqref{boundDBM-RS} holds true.
\end{corollary}

\begin{proof}
Let $q$ be the unique solution of \eqref{RSconsGauss}.
Let $a\in\R_+^{K-1}$ verifying \eqref{ceA}, so that the relation \eqref{theta-M} holds true.
Then \eqref{ATDBM} and \eqref{RSconsGauss} rewrite respectively as:
\be \begin{cases}
\theta_p(a)^2\;\E\cosh^{-4}\left(z\,\sqrt{2\,q_p\,\theta_p(a)^2+v_p}\,\right) \,\leq\, \dfrac{1}{2} \\[8pt]
q_p \,=\, \E\tanh^{2}\left(z\,\sqrt{2\,q_p\,\theta_p(a)^2+v_p}\,\right)
\end{cases}\ee
for every $p=1,\dots,K\,$.
By Chen's result \eqref{ATGauss}, this entails
\be 
\pSK\big(\theta_p(a),h^{(p)}\big)\,=\,\pSKRS\big(q_p\,;\theta_p(a),h^{(p)}\big)
\ee
for every $p=1,\dots,K\,$.
Therefore by Theorem \ref{teo:RS}, 
\be 
\PDBM(a;\,\beta,\lambda,h) \,=\, \pRS(q;\,\beta,\lambda,h)
\ee
and the bound \eqref{boundDBM-RS} holds true.
\end{proof}


\appendix
\section{Matching polynomials} \label{sec:appendix}
In this Appendix we give some properties of the polynomials $\DDelta_p(x,t)$ introduced by Definition \ref{def:Liebpoly} and characterizing the annealed region of the DBM. 
In particular we are interested in the location of the zeros of $\DDelta_p$, namely the points $x\in\C$ such that $\DDelta_p(x,t)=0\,$.

Theorem \ref{thm:HL} and Corollary \ref{cor:HL} are due to Heilmann and Lieb \cite{HL} and show that the zeros are real and have an interlacing proprety.
Proposition \ref{prop:zeros} and Corollary \ref{cor:zeros}, by using these results, contribute to the proof of Proposition \ref{annp}. Precisely we show that the zeros of $\DDelta_K$ lie in the interval $(-\rho,\rho)$ if and only if all the polynomials $\DDelta_p$ for $p\leq K$ are positive at $x=\rho\,$.

\begin{theorem}[Heilmann-Lieb \cite{HL}] \label{thm:HL}
Let $\tp>0$ for all $p=1,\dots,K-1\,$. Then for every $p=1,\dots,K$
\begin{itemize}
\item[i)] the zeros of $\DDelta_p$ are real and simple;
\item[ii)] if $p\geq1$, the zeros of $\DDelta_p$ ``interlace'' with those of $\DDelta_{p-1}$. Namely, denoting by $x_1^{(p-1)}<\dots<x_{p-1}^{(p-1)}$ the zeros of $\DDelta_{p-1}$ and by $x_1^{(p)}<\dots<x_{p}^{(p)}$ the zeros of $\DDelta_p$, we have:
\be \label{eq:interlace}
x_1^{(p)} < x_1^{(p-1)} < x_2^{(p)} < x_2^{(p-1)} < \dots < x_{p-1}^{(p)} < x_{p-1}^{(p-1)} < x_p^{(p)}\ .
\ee
\end{itemize}
\end{theorem}

\begin{proof}
The statement is trivially true for $p=0$ and $p=1$. Consider $p\geq1$, assume the statement holds true for $p-1$ and $p$ and prove it for $p+1$.
By induction hypothesis the zeros of $\DDelta_{p}$ and those of $\DDelta_{p-1}$ are real and simple and they are interlaced, namely \eqref{eq:interlace} holds true.

Since the zeros of $\DDelta_{p-1}$ are simple, $\DDelta_{p-1}$ changes its sign exactly at every $x_1^{(p-1)},\dots,x_{p-1}^{(p-1)}$. By \eqref{eq:interlace}, it follows that $\DDelta_{p-1}$ has alternating signs at the points $x_1^{(p)},\dots,x_{p}^{(p)}$.
Therefore also $\DDelta_{p+1}$ has alternating signs at the points $x_1^{(p)},\dots,x_{p}^{(p)}\,$, indeed by the recursion relation \eqref{eq:Liebpoly}
\be \label{eq:zeros}
\DDelta_{p+1}\!\big(x_k^{(p)},t\big) \,=\, -\,\underbrace{\tp}_{>0}\,\DDelta_{p-1}\!\big(x_k^{(p)},t\big)
\ee
for every $k=1,\dots,p$.
As a consequence $\DDelta_{p+1}$ has (at least) one zero in each interval $\big(x_k^{(p)},\,x_{k+1}^{(p)}\big)$ for $k=1,\dots,p-1$.
Moreover, since $\DDelta_{p+1}$ and $\DDelta_{p-1}$ share the same sign as $x\to\infty$ and as $x\to-\infty\,$, \eqref{eq:zeros} implies that $\DDelta_{p+1}$ has (at least) one zero in $\big(x_p^{(p)},\,\infty\big)$ and (at least) one zero in $\big(-\infty,\,x_{1}^{(p)}\big)\,$.
Since the zeros of $\DDelta_{p+1}$ are exactly $p+1$, the thesis follows.
\end{proof}

Theorem \ref{thm:HL} can be extended to the case of non-negative coefficients.

\begin{corollary}[Heilmann-Lieb \cite{HL}] \label{cor:HL}
Let $\tp\geq0$ for all $p=1,\dots,K-1\,$. Then for every $p=1,\dots,K$
\begin{itemize}
\item[i)] the zeros of $\DDelta_p$ are real;
\item[ii)] if $p\geq1$, the zeros of $\DDelta_p$ ``weakly interlace'' with those of $\DDelta_{p-1}$. Namely, denoting by $x_1^{(p-1)}\leq\dots\leq x_{p-1}^{(p-1)}$ the zeros of $\DDelta_{p-1}$ and by $x_1^{(p)}\leq \dots\leq x_{p}^{(p)}$ the zeros of $\DDelta_p$ repeated according to their multiplicity, we have:
\be 
x_1^{(p)} \leq x_1^{(p-1)} \leq x_2^{(p)} \leq x_2^{(p-1)} \leq \dots \leq x_{p-1}^{(p)} \leq x_{p-1}^{(p-1)} \leq x_p^{(p)}\ .
\ee
\end{itemize}
\end{corollary}

\begin{proof}
It follows from Theorem \ref{thm:HL} by continuity.
\end{proof}

\begin{remark} \label{symmzeros}
The zeros of $\DDelta_p$ are symmetric with respect to $x=0$. Indeed
\be 
\DDelta_p(x,t) \,=\, (-1)^p\,\DDelta_p(-x,t)
\ee
because both polynomials verify the same recursion relation \eqref{eq:Liebpoly}.
\end{remark}

\begin{proposition} \label{prop:zeros}
Let $\tp>0$ for all $p=1,\dots,K-1\,$. Then for every $\rho>0$ the followings are equivalent:
\begin{itemize}
\item[i)] the zeros of $\DDelta_K$ are contained in $(-\rho,\rho)\,$;
\item[ii)] the zeros of $\DDelta_p$ are contained in $(-\rho,\rho)$ for every $p=1,\dots,K\,$;
\item[iii)] $\DDelta_p(\rho,t)>0$ for every $p\leq K$ such that $p\equiv_{\textrm{mod}2}K\,$;
\item[iv)] $\DDelta_p(\rho,t)>0$ for every $p=1,\dots,K\,$.
\end{itemize}
\end{proposition}

\begin{proof}
i$\Rightarrow$ii. This is a consequence of Theorem \ref{thm:HL}.

ii$\Rightarrow$iii. Trivial since $\DDelta_p(x,t)\to\infty$ as $x\to\infty$ for every $p\geq1\,$.

iii$\Rightarrow$iv. From the recursion relation \eqref{eq:Liebpoly}, one sees that if $\DDelta_{p+1}(\rho,t)>0$ and $\DDelta_{p-1}(\rho,t)>0$ then also $\DDelta_{p}(\rho,t)>0\,$.

iv$\Rightarrow$i. By contradiction assume that $\DDelta_p(\rho,t)>0$ for every $p=1,\dots,K$ and not all the zeros of $\DDelta_K$ are contained in $(-\rho,\rho)$.\\
\textit{Claim:} $\DDelta_p$ has at least two zeros in $(\rho,\infty)$ for every $p=2,\dots,K\,$.\\
We are going to prove the claim by induction. It will contradict the fact that $\DDelta_2$ has only one positive zero.\\
Let's start from $p=K$. By hypothesis $\DDelta_K(\rho,t)>0$ and
$\DDelta_K$ has a zero $x_0^{(K)}\in(\rho,\infty)\,$.
Theorem \ref{thm:HL} guarantees that $\DDelta_K$ changes its sign at $x=x_0^{(K)}$ (because every zero is simple).
On the other hand we know that $\DDelta_K(x,t)\to\infty$ as $x\to\infty$.
Therefore $\DDelta_K$ has (at least) another zero $x_1^{(K)}\in(\rho,\infty)\,$, $x_1^{(K)}\neq x_0^{(K)}\,$. This proves the claim for $p=K\,$.\\
Now let $p\leq K$, assume the claim for $p$ and prove it for $p-1\,$. By induction hypothesis 
$\DDelta_p$ has two zeros $x_0^{(p)},x_1^{(p)}\in(\rho,\infty)\,$, $x_1^{(p)}\neq x_0^{(p)}\,$.
By Theorem \ref{thm:HL} it follows that
$\DDelta_{p-1}$ has a zero $x_0^{(p-1)}\in(\rho,\infty)$ (interlacing of the zeros).
Since by hypothesis $\DDelta_{p-1}(\rho,t)>0$ and $\DDelta_{p-1}(x,t)\to\infty$ as $x\to\infty$, it follows that
$\DDelta_{p-1}$ has another zero $x_1^{(p-1)}\in(\rho,\infty)\,$, $x_1^{(p-1)}\neq x_0^{(p-1)}$.
\end{proof}

Also Proposition \ref{prop:zeros} extends to the case of non-negative coefficients.
\begin{corollary} \label{cor:zeros}
Let $\tp\geq0$ for all $p=1,\dots,K-1\,$. Then for every $\rho>0$ the followings are equivalent:
\begin{itemize}
\item[i)] the zeros of $\DDelta_K$ are contained in $(-\rho,\rho)\,$;
\item[ii)] the zeros of $\DDelta_p$ are contained in $(-\rho,\rho)$ for every $p=1,\dots,K\,$;
\item[iii)] $\DDelta_p(\rho,t)>0$ for every $p\leq K$ such that $p\equiv_{\textrm{mod}2}K\,$;
\item[iv)] $\DDelta_p(\rho,t)>0$ for every $p=1,\dots,K\,$.
\end{itemize}
\end{corollary}

\begin{proof}
Implications i$\Rightarrow$ii$\Rightarrow$iii$\Rightarrow$iv are proven as before.
iv$\Rightarrow$i follows from Proposition \ref{prop:zeros} by continuity.
\end{proof}

\section*{Acknowledgements}
The authors thank Adriano Barra, Wei-Kuo Chen, Francesco Guerra and Daniele Tantari for interesting discussions.
D.A. is grateful to Alberto Viscardi for his contribution to Proposition \ref{prop:rhobound}.
P.C. was partially supported by PRIN project Statistical Mechanics and Complexity (2015K7KK8L).
D.A. and E.M. were partially supported by Progetto Almaidea 2018.

\end{document}